\documentclass[aps,preprint,showpacs,groupedaddress]{revtex4-1}
\usepackage{amssymb}
\usepackage{mathrsfs}
\usepackage{amsmath}
\usepackage{mathtools}
\usepackage{pstricks}
\usepackage{color}
\usepackage{graphicx}
\usepackage{amsthm}
\newtheorem*{theorem}{Theorem}
\usepackage{physics}

\begin{document}


\title{\bf Probing the two-scale-factor universality hypothesis by exact rotation symmetry-breaking mechanism}



\author{J. F. S. Neto}
\author{K. A. L. Lima}
\author{P. R. S. Carvalho}
\email{prscarvalho@ufpi.edu.br}
\affiliation{\it Departamento de F\'\i sica, Universidade Federal do Piau\'\i, 64049-550, Teresina, PI, Brazil}

\author{M. I. Sena-Junior}
\email{marconesena@poli.br}
\affiliation{\it Escola Polit\'{e}cnica de Pernambuco, Universidade de Pernambuco, 50720-001, Recife, PE, Brazil}
\affiliation{\it Instituto de F\'{i}sica, Universidade Federal de Alagoas, 57072-900, Macei\'{o}, AL, Brazil}




\begin{abstract}
We probe the two-scale factor universality hypothesis by evaluating, firstly explicitly and analytically at the one-loop order, the loop quantum corrections to the amplitude ratios for O($N$) $\lambda\phi^{4}$ scalar field theories with rotation symmetry-breaking in three distinct and independent methods in which the rotation symmetry-breaking mechanism is treated exactly. We show that the rotation symmetry-breaking amplitude ratios turn out to be identical in the three methods and equal to their respective rotation symmetry-breaking ones, although the amplitudes themselves, in general, depend on the method employed and on the rotation symmetry-breaking parameter. At the end, we show that all these results can be generalized, through an inductive process based on a general theorem emerging from the exact calculation, to any loop level and physically interpreted based on symmetry ideas.   
\end{abstract}


\maketitle


\section{Introduction} 

\par The identical critical behavior displayed by different physical systems, as a fluid and a ferromagnet, near a continuous phase transition, had lead to the genesis of the universality concept \cite{PhysRevLett.16.11,JPhysC.2.1883,JPhysC.2.2158,PhysRev.176.738,PhysRevLett.24.1479,Kadanoff} related to the scaling hypothesis \cite{JChemPhys.43.3898,Phys.2.263,ZhEkspTeorFiz.50.439,JChemPhys.39.842}. The critical behavior of such systems is characterized by an identical set of critical exponents. When the critical behavior of two or more systems is described by equal critical exponents, we say that they belong to the same universality class \cite{PhysRevLett.28.240,PhysRevLett.28.548,Wilson197475}. This occurs when they share the same dimension $d$, $N$ and symmetry of some $N$-component order parameter and if the interactions present are of short- or long-range type. We will deal with the general O($N$) universality class which is a generalization of the specific models with short-range interactions: Ising ($N=1$), XY ($N=2$), Heisenberg ($N=3$), self-avoiding random walk ($N=0$), spherical ($N \rightarrow \infty$) etc \cite{Pelissetto2002549}. Furthermore, different systems can be represented by a single universal equation of state, once one has fixed two independent thermodynamic scales, as the order parameter and its conjugate field scales. Then, the equation of state and amplitude ratios for the thermodynamic functions are universal and thus satisfy the thermodynamic universality hypothesis \cite{PhysRev.24.1479,PhysRevB.12.1947}. Stauffer, Ferer and Wortis \cite{PhysRevLett.29.345} generalized the thermodynamic universality concept to the \textit{two-scale-factor} universality hypothesis for correlation functions, where before that work, it was suggested that universality for correlation functions would be inferred after the choice of three scales, with the additional scale to the thermodynamic ones being the length scale. This hypothesis asserts that, near the critical point, the length scale is not independent and it is related to the thermodynamic scales. Thus the universal correlation function can be fully determined after the choice of just two independent scales. Unlike the critical exponents themselves, the critical amplitudes of the thermodynamic  and correlation functions, near the critical point, are not universal quantities. The universal quantities in this case are some amplitude ratios of them. The aim of this work is to evaluate these amplitude ratios. 

\par The purpose of this paper is to employ field-theoretic renormalization-group and $\epsilon$-expansion techniques for computing, firstly explicitly and analytically at the one-loop order, the loop quantum corrections to the amplitude ratios for rotation symmetry-breaking O($N$) $\lambda\phi^{4}$ scalar field theories. This task plays a similar role in the description of a given universality class, although the evaluation of amplitude ratios is, in general, harder than for critical exponents. There is another field-theoretic renormalization-group approach for evaluating critical quantities. It is called the field-theoretic renormalization-group at fixed-dimension approach \cite{JStatPhys2349,PhysRevB353585} and is based on the computation of critical quantities directly in three dimensions. One important application of a rotation symmetry-breaking scalar field theory on the research area of high energy physics is explaining the Higgs behavior through the recently proposed rotation symmetry-breaking Higgs sector of the extended standard model \cite{PhysRevD.84.065030,Carvalho2013850,Carvalho2014320}. In a conventional rotation-invariant theory, the critical amplitudes are the amplitudes of the scaling thermodynamic functions and correlation functions, defined above and below the critical temperature. These functions, in turn, are a result of some derivative operations, some of them being derivatives of the magnetization $M$, of the reduced temperature, i. e., a parameter that is proportional to the difference between some arbitrary temperature and the critical one $t \propto T - T_{c}$ etc with respect to the free energy density or, in the present language, the effective potential with spontaneous symmetry breaking. The effective potential at the loop level considered explicitly and analytically here, the one-loop order, is composed of two terms. The first term is responsible for the so called Landau approximation values to the amplitude ratios, valid for $d\geq 4$. In the Landau regime, the fluctuations of the scalar field $\phi$, whose mean value is identified to the magnetization of the system are discarded. The second one, representing corrections to the Landau approximation, which takes into account the fluctuations as loop quantum corrections, valid for $2<d<4$, is the infinite sum of all the one-loop $1$PI vertex parts with amputated external legs. Initially, the effective potential is written in its bare or nonrenormalized form, thus plagued by infrared divergences, typical for massless theories, as the treated in this work. These divergences must be removed of the theory and are contained in just a few $1$PI vertex parts, the $\Gamma_{B}^{(2)}$, $\Gamma_{B}^{(4)}$ and $\Gamma_{B}^{(2,1)}$ functions commonly called primitively divergent $1$PI vertex parts. All the others divergent $1$PI vertex parts, obtained through a skeleton expansion \cite{ZinnJustin} of the former, turn out to be automatically renormalized once one has renormalized the primitive divergent ones. For attaining the renormalized theory, we will apply three independent renormalization schemes: normalization conditions \cite{BrezinLeGuillouZinnJustin}, the minimal subtraction scheme \cite{Amit} and the Bogoliubov-Parasyuk-Hepp-Zimmermann (BPHZ) methods \cite{BogoliubovParasyuk,Hepp,Zimmermann}. Universality is satisfied if the final results for the amplitude ratios, in the three distinct methods, are identical, although the amplitudes themselves, in general, depend on the renormalization scheme employed and on the rotation symmetry-breaking mechanism through the introduction of an appropriate rotation symmetry-breaking parameter to be defined below. As universality means that the critical exponents do not depend on the system in a given universality class, the corresponding critical exponents must be the same when obtained through any renormalization scheme. Universality then arises in any renormalization scheme through the flow of the renormalized coupling constant to its fixed point value in which scale invariance is manifest. The present calculation of amplitude ratios in just one renormalization scheme would be enough for showing that they are independent of the symmetry breaking tensor $K_{\mu\nu}$ at one-loop order. But similar calculations in different renormalization schemes besides providing a check of the final results that must be identical, give more robustness on the \emph{two-scale-factor} universality hypothesis validity. Furthermore, the minimal subtraction method for obtaining the amplitude ratios presented here, as opposed to the normalization conditions method, is not found in the literature. The probing of a possible effect of the rotation symmetry-breaking mechanism on the universality properties of the systems studied here, starts with the introduction into the rotation-invariant standard theory, the kinetic rotation symmetry-breaking O($N$) operator $K_{\mu\nu}\partial^{\mu}\phi\partial^{\nu}\phi$, as introduced for the pioneer evaluation of rotation symmetry-breaking critical exponents by one of the authors and co-workers \cite{EurophysLett.108.21001,Int.J.Geom.MethodsMod.Phys.13.1650049,Int.J.Mod.Phys.B.30.1550259}, although non-exactly in the rotation symmetry-breaking mechanism through tedious calculations in powers of $K_{\mu\nu}$. The dimensionless, symmetric, constant rotation symmetry-breaking coefficients $K_{\mu\nu} = K_{\nu\mu}$ are equal for all the $N$ components of the field and leave intact the O($N$) symmetry of the $N$-component field. Physically, they act as a constant background field. If the rotation symmetry-breaking coefficients are kept at arbitrary values, the rotation symmetry symmetry is violated if these coefficients do not transform as a second order tensor under rotation transformations. As for the earlier works on the computation of rotation symmetry-breaking critical exponents, the rotation symmetry-breaking theory can be used for studying the symmetry aspects of the O($N$) two-scale-factor universality class in the rotation symmetry-breaking scenario, now treating the rotation symmetry-breaking mechanism exactly.

\section{Amplitude ratios in normalization conditions scheme}

\par The normalization conditions renormalization scheme is characterized by fixing the external momenta of the primitively divergent $1$PI vertex parts at a nonzero value scaled by some arbitrary momentum scale $\kappa$, at the symmetry points $SP$ and $\overline{SP}$
\begin{eqnarray}\label{ygfdxzsze}
\Gamma^{(2)}(P^{2} + K_{\mu\nu}P^{\mu}P^{\nu} = 0, g) = 0, 
\end{eqnarray}
\begin{eqnarray}
\frac{\partial \Gamma^{(2)}(P^{2} + K_{\mu\nu}P^{\mu}P^{\nu}, g)}{\partial (P^{2} + K_{\mu\nu}P^{\mu}P^{\nu})}\Biggr|_{P^{2} + K_{\mu\nu}P^{\mu}P^{\nu} = \kappa^{2}}   = 1,
\end{eqnarray}
\begin{eqnarray}
\Gamma^{(4)}(P^{2} + K_{\mu\nu}P^{\mu}P^{\nu}, g)|_{SP} = g, 
\end{eqnarray}
\begin{eqnarray}\label{jijhygtfrd}
\Gamma^{(2,1)}(P_{1}, P_{2}, Q_{3}, g)|_{\overline{SP}} = 1,
\end{eqnarray}
where for SP: $P_{i}\cdot P_{j} = (\kappa^{2}/4)(4\delta_{ij}-1)$, implying that $(P_{i} + P_{j})^{2} \equiv P^{2} = \kappa^{2}$ for $i\neq j$ and for $\overline{SP}$: $P_{i}^{2} = 3\kappa^{2}/4$ and $P_{1}\cdot P_{2} = -\kappa^{2}/4$, implying $(P_{1} + P_{2})^{2} \equiv P^{2} = \kappa^{2}$, of the multiplicatively renormalized primitively $1$PI vertex parts $\Gamma^{(n, l)}(P_{i}, Q_{j}, g, \kappa) = Z_{\phi}^{n/2}Z_{\phi^{2}}^{l}\Gamma_{B}^{(n, l)}(P_{i}, Q_{j}, \lambda_{0})$ ($i = 1, \cdots, n$, $j = 1, \cdots, l$, where for $(n, l)\neq(0, 2)$, the function $\Gamma_{B}^{(0, 2)}$ is renormalized additively), generated by the initially bare rotation symmetry-breaking Lagrangian density
\begin{eqnarray}\label{huytrji}
\mathcal{L}_{B} = \frac{1}{2}\partial^{\mu}\phi_{B}\partial_{\mu}\phi_{B} + K_{\mu\nu}\partial^{\mu}\phi_{B}\partial^{\nu}\phi_{B} + \frac{\lambda_{B}}{4!}\phi_{B}^{4} + \frac{1}{2}t_{B}\phi_{B}^{2},
\end{eqnarray}
where the conditions (\ref{ygfdxzsze})-(\ref{jijhygtfrd}) permit us to renormalize the bare field $\phi_{B}$, coupling constant $\lambda_{B}$ and composite field coupling constant $t_{B}$ parameters. Thus after the renormalization of these parameters, we can write down the renormalized rotation symmetry-breaking free energy density at the fixed point with spontaneous symmetry breaking at one-loop level as \cite{BrezinLeGuillouZinnJustin}
\begin{eqnarray}\label{fhdldglkjdflk}
&& \mathcal{F}(t,M,g^{\ast}) = \frac{1}{2}tM^{2} + \frac{1}{4!}g^{\ast}M^{4} + \frac{1}{4}\left[ Nt^{2} + \frac{N+2}{3}tg^{\ast}M^{2} + \frac{N+8}{36}(g^{\ast}M^{2})^{2} \right]\parbox{10mm}{\includegraphics[scale=1.0]{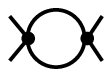}}_{SP} + \nonumber \\ && \frac{1}{2}\int d^{d}q\Biggl[ \ln\left( 1 + \frac{t + g^{\ast}M^{2}/2}{q^{2} + K_{\mu\nu}q^{\mu}q^{\nu}} \right) + (N-1)\ln\left( 1 + \frac{t + g^{\ast}M^{2}/6}{q^{2} + K_{\mu\nu}q^{\mu}q^{\nu}} \right) - \frac{N+2}{6}\frac{tg^{\ast}M^{2}}{q^{2} + K_{\mu\nu}q^{\mu}q^{\nu}} \Biggl],\nonumber \\
\end{eqnarray}
where $g$, $t$ and $M$ are the renormalized coupling constant, composite field coupling constant and magnetization (as being the renormalized nonzero field mean value $M = \langle\phi\rangle$ in the spontaneously broken direction), respectively. The coupling constant $g^{\ast}$ is the fixed point of the theory, the value for which the renormalized coupling constant flows naturally when the renormalized theory is attained and, in general, is obtained as the nontrivial root of the $\beta$-function of the respective theory in the respective renormalization scheme. The ``fish" diagram $\parbox{6mm}{\includegraphics[scale=0.6]{fig10.eps}}_{SP}$, whose internal line is given by the massless propagator $\parbox{7mm}{\includegraphics[scale=0.6]{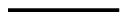}}^{-1} \equiv q^{2} + K_{\mu\nu}q^{\mu}q^{\nu}$, is evaluated at the symmetry point $SP$ after we set $\kappa^{2} = 1$, because we can redefine all momenta of the diagrams in units of $\kappa$. Thus the redefined external momenta turn out to be dimensionless and the symmetry point now is given by $P^{2} + K_{\mu\nu}P^{\mu}P^{\nu} = \kappa^{2} \rightarrow 1$. Then we absorb the dependence on $\kappa$ of the diagrams into the coupling constant. The ``fish" diagram was evaluated in a expansion in the dimensional regularization parameter $\epsilon = 4 - d$ and exactly in $K$ \cite{CarvalhoSenaJunior}, see Sect. \ref{All-loop order amplitude ratios}. In the analytical computation of the momentum integral, we apply the known definition of ref. \cite{Amit} in which the $d$-dimensional surface area factor is absorbed into a redefinition of the coupling constant, since each loop momentum integration is accompanied of this factor. Thus, we can write
\begin{eqnarray}
\parbox{10mm}{\includegraphics[scale=1.0]{fig10.eps}}_{SP} = \frac{1}{\epsilon}\left(1 + \frac{1}{2}\epsilon \right)\mathbf{\Pi},
\end{eqnarray} 
where $\mathbf{\Pi} = 1/\sqrt{det(\mathbb{I} + \mathbb{K})}$ is a rotation symmetry-breaking full factor emerging from the exact calculation where we change the variables \cite{CarvalhoSenaJunior} through coordinates redefinition applied in momentum space directly in Feynman diagrams $q^{\prime} = \sqrt{\mathbb{I} + \mathbb{K}}\hspace{1mm}q$ resulting in the fact that each loop integration is accompanied by a rotation symmetry-breaking full $\mathbf{\Pi}$ factor and according to \cite{CarvalhoSenaJunior},
\begin{eqnarray}\label{hufhufghuf}
g^{\ast} = \frac{6\epsilon}{(N + 8)\mathbf{\Pi}}\left\{ 1 + \epsilon\left[ \frac{(9N + 42)}{(N + 8)^{2}} -\frac{1}{2} \right]\right\}. 
\end{eqnarray}
Now, we are in a position to evaluate the critical amplitudes of the thermodynamic and correlation functions. The existence of two independent scales, leads naturally to the existence of ten relations among the twelve critical exponents $\alpha, \alpha^{\prime}, \gamma, \gamma^{\prime}, \nu, \nu^{\prime}, \beta, \delta, \eta, \alpha_{c}, \gamma_{c}, \nu_{c}$ as well as for the critical amplitudes, with ten universal relations among the twelve critical amplitudes
\begin{eqnarray}\label{uhdfuhfuh}
\begin{array}{lcr}
\mbox{\textrm{Critical isochore: $T > T_{c}$, $H = 0$}} &  &  \\
\mbox{} \xi = \xi_{0}^{+}t^{-\nu}, \chi = C^{+}t^{-\gamma}, C_{s} = \frac{A^{+}}{\alpha_{+}}t^{-\alpha} &  \\ [10pt]
\mbox{\textrm{Critical isochore: $T < T_{c}$, $H = 0$}} &  &  \\
\mbox{} \xi = \xi_{0}^{-}t^{-\nu}, \chi = C^{-}t^{-\gamma}, C_{s} = \frac{A^{-}}{\alpha_{-}}t^{-\alpha}, M = B(-t)^{\beta} &  \\ [10pt]
\mbox{\textrm{Critical isotherm: $T = T_{c}$, $H \neq 0$}} &  &  \\ 
\mbox{} \xi = \xi_{0}^{c}|H|^{-\nu_{c}}, \chi = C^{c}|H|^{-\gamma_{c}}, C_{s} = \frac{A^{c}}{\alpha_{c}}|H|^{-\alpha_{c}}, H = DM^{\delta} &  \\ [10pt]
\mbox{\textrm{Critical point: $T = T_{c}$, $H = 0$}} &  &  \\
\mbox{} \chi(p) = \widehat{D}p^{\eta-2}. &  \\ \end{array} \nonumber
\end{eqnarray}\label{jlkfdjkglk}
Fortunately, not all critical exponents must be evaluated, because not all of them are independent. Some of them are related, as the ones defined above and below the critical temperature and through some scaling relations among them, namely: $\alpha = \alpha^{\prime}, \gamma = \gamma^{\prime}, \nu = \nu^{\prime}, \gamma = \beta(\delta - 1), \alpha = 2 - 2\beta - \gamma, 2 - \alpha = d\nu, \gamma = (2 - \eta)\nu, \alpha_{c} = \alpha/\beta\delta, \gamma_{c} = 1 - 1/\delta, \nu_{c} = \nu/\beta\delta$, thus remaining two independent ones. Also, as $H$ and $\chi$ are related on the critical isotherm, the universal relation $\delta C^{c}D^{1/\delta}$ follows, and nine universal relations among the critical amplitudes remains. More than nine universal relations can be derived, some of them being dependent of a minimal set of nine ones. This shows that we can choice a given minimal set. The minimal set chosen in this paper will be that whose the $\epsilon$-expansion results are displayed in ref. \cite{PrivmanHohenbergAharony}, originally evaluated in refs. therein. 

\par \textit{Equation of state}. Before computing the amplitude ratios themselves, it is important to begin by computing the equation of state and its universal form. The equation of state is obtained as a first derivative of the free energy with respect to the magnetization, namely $H = \partial \mathcal{F}/\partial M$ \cite{BrezinLeGuillouZinnJustin}, whose diagrammatic expression is given by \cite{Bervillier}
\begin{eqnarray}\label{jglkfjkflj}
&& H/M = t + \frac{1}{6}g^{\ast}M^{2} + \frac{1}{2}g^{\ast}\Biggl\{ \Biggl[ \parbox{12mm}{\includegraphics[scale=1.0]{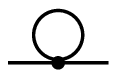}}_{(1)} + (t + g^{\ast}M^{2}/2)\parbox{10mm}{\includegraphics[scale=1.0]{fig10.eps}}_{SP} \Biggl]  + \nonumber \\ &&  \frac{N - 1}{3}\Biggl[ \parbox{12mm}{\includegraphics[scale=1.0]{fig1.eps}}_{(N - 1)} + (t + g^{\ast}M^{2}/6)\parbox{10mm}{\includegraphics[scale=1.0]{fig10.eps}}_{SP} \Biggl] \Biggl\}. 
\end{eqnarray}
The indices $(1)$ and $(N-1)$ indicate that in the respective diagram, the internal propagators are $\parbox{7mm}{\includegraphics[scale=0.6]{fig9.eps}}_{(1)}^{-1} \equiv q^{2} + K_{\mu\nu}q^{\mu}q^{\nu} + g^{\ast}M^{2}/2$ and $\parbox{7mm}{\includegraphics[scale=0.6]{fig9.eps}}_{(N - 1)}^{-1} \equiv q^{2} + K_{\mu\nu}q^{\mu}q^{\nu} + g^{\ast}M^{2}/6$, respectively. The evaluation of the Feynman diagrams $\parbox{7mm}{\includegraphics[scale=.6]{fig1.eps}}_{(1)}$ and $\parbox{7mm}{\includegraphics[scale=.6]{fig1.eps}}_{(N - 1)}$ in $d = 4 - \epsilon$ result 
\begin{eqnarray}\label{fldsjlkfdjlk}
\parbox{12mm}{\includegraphics[scale=1.0]{fig1.eps}}_{(1)} = -\frac{t + g^{\ast}M^{2}/2}{\epsilon}\Biggl[ 1 - \frac{1}{2}\ln \Biggl( t + \frac{g^{\ast}M^{2}}{2} \Biggl)\epsilon \Biggl]\mathbf{\Pi},  
\end{eqnarray}
and a similar expression for $\parbox{7mm}{\includegraphics[scale=.6]{fig1.eps}}_{(N - 1)}$ with $g^{\ast}M^{2}/2 \rightarrow g^{\ast}M^{2}/6$. The equation of state is clearly a rotation symmetry-breaking one. But if $H(x)$ is normalized at the values $x = 0$ and $x = -1$ such that $H = M^{\delta}$ and $H = 0$, respectively, where $x = t(g^{\ast}M^{2})^{-1/2\beta}$, the rotation symmetry-breaking $\Pi$ and $f^{(2)}$ factors disappear and it assumes its known universal form \cite{BrezinWallaceWilson}. Now we proceed to evaluate the minimal set of amplitude ratios.  
 
\par $A^{+}/A^{-}$. The critical amplitudes for the specific heat can be obtained as the second derivative of the free energy $\mathcal{F}(t,M,g^{\ast})$ with respect to $t$. For describing the regions above and below the transition point, the values of the magnetization in the respective regions can calculated by minimizing the effective potential (equivalently the roots of $H$), giving the magnetization values of the O($N$) symmetric phase and the spontaneously broken one. The results for the referred amplitudes are
\begin{eqnarray}\label{hduhsduhsdu}
&& A^{+} = \frac{N}{4}\Bigg[ 1 + \Bigg( \frac{4}{4 - N} + A_{N} \Bigg)\epsilon \Bigg]\mathbf{\Pi}, 
\end{eqnarray}
\begin{eqnarray}\label{jdoijdio}
 A^{-} = \Bigg[ 1 + \Bigg( \frac{N}{4 - N} - \frac{4 - N}{2(N + 8)}\ln 2 + A_{N} \Bigg)\epsilon \Bigg]\mathbf{\Pi} ,
\end{eqnarray}     
where
\begin{eqnarray}\label{hduhsduhsduu}
A_{N} = \frac{1}{2} - \frac{9N + 42}{(N + 8)^{2}} - \frac{4 - N}{(N + 8)} - \frac{(N + 2)(N^{2} + 30N + 56)}{2(4 - N)(N + 8)^{2}}. 
\end{eqnarray}

\par $C^{+}/C^{-}$. We can obtain the amplitude ratio $C^{+}/C^{-}$ by computing the susceptibility in terms of the effective potential, or equivalently of the equation of state, through $\chi^{-1} = \partial^{2}\mathcal{F}/\partial M^{2} = \partial H/\partial M$ and evaluating the amplitudes above and below the transition. We have to mention a peculiarity here: Below the critical temperature, the susceptibility is defined only for Ising systems due to the presence of Goldstone modes. Thus we have  
\begin{eqnarray}\label{jhdshudfhuifh}
&& C^{+} = 1 - \frac{N + 2}{2(N + 8)}\epsilon , 
\end{eqnarray}
\begin{eqnarray}\label{ygdsuygsduygds}
&& C^{-} = \frac{1}{2}\Bigg[1 - \frac{1}{6}\left( 4 + \ln 2 \right)\epsilon\Bigg].
\end{eqnarray}     

\par $Q_{1}$. The amplitude ratio $Q_{1}$ is related to the $R_{\chi}$ one through $Q_{1} = R_{\chi}^{-1/\delta}$, where $R_{\chi}$ is defined as $R_{\chi} = C^{+}DB^{\delta - 1}$. Thus the universality of $Q_{1}$ is ensured if $R_{\chi}$ is universal. The amplitude $C^{+}$ is displayed in eq. (\ref{jhdshudfhuifh}). The amplitude $D$ is obtained by normalizing $H(x)$ only at the value $x = 0$ such that $H = M^{\delta}$, where $x = t(g^{\ast}M^{2})^{-1/2\beta}$, implying that
\begin{eqnarray}\label{fhduhdfu}
D = \frac{1}{6}g^{\ast(\delta - 1)/2}\Bigg[1 + \frac{1}{2}\Bigg( 1 - \ln 2 - \frac{N - 1}{N + 8}\ln 3 \Bigg)\epsilon\Bigg]. 
\end{eqnarray} 
The amplitude B can be calculated from the nonzero root of $H$, namely
\begin{eqnarray}\label{huhfuhfuhf}
B = \Bigg\{\frac{N + 8}{\epsilon}\Bigg[1 - \frac{3}{N + 8}\left( 1 + \ln 2 \right)\epsilon - \Bigg(\frac{9N + 42}{(N + 8)^{2}} -\frac{1}{2} \Bigg)\epsilon \Bigg]\mathbf{\Pi}\Bigg\}^{1/2}. 
\end{eqnarray} 

\par $R_{c}$. The ratio $R_{c}$ is defined as $R_{c} = A^{+}C^{+}/B^{2}$. All the amplitudes necessary to the computation of $R_{c}$ were evaluated already in eqs. (\ref{hduhsduhsdu}), (\ref{jhdshudfhuifh}) and (\ref{huhfuhfuhf})

\par $\xi_{0}^{+}/\xi_{0}^{-}$. For calculating the amplitude ratio between the correlation length above and below the transition, we have to consider the momentum-dependent longitudinal correlation function \cite{BrezinLeGuillouZinnJustin} with the diagrammatic expansion \cite{Bervillier}
\begin{eqnarray}\label{uhduhuhg}
&& \Gamma_{L}(P^{2} + K_{\mu\nu}P^{\mu}P^{\nu},t,M) = P^{2} + K_{\mu\nu}P^{\mu}P^{\nu} + t + \frac{1}{2}g^{\ast}M^{2} + \nonumber \\ && \frac{1}{2}g^{\ast}\Biggl\{ \Biggl[ \parbox{12mm}{\includegraphics[scale=1.0]{fig1.eps}}_{(1)} + (t + g^{\ast}M^{2}/2)\parbox{10mm}{\includegraphics[scale=1.0]{fig10.eps}}_{SP} \Biggl]  +  \frac{N - 1}{3}\Biggl[ \parbox{12mm}{\includegraphics[scale=1.0]{fig1.eps}}_{(N - 1)} + (t + g^{\ast}M^{2}/6)\parbox{10mm}{\includegraphics[scale=1.0]{fig10.eps}}_{SP} \Biggl] + \nonumber \\ && g^{\ast}M^{2}\Biggl[ \Biggl(\parbox{10mm}{\includegraphics[scale=1.0]{fig10.eps}}_{SP} - \parbox{10mm}{\includegraphics[scale=1.0]{fig10.eps}}_{(1)} \Biggl) + \frac{N - 1}{9}\Biggl(\parbox{10mm}{\includegraphics[scale=1.0]{fig10.eps}}_{SP} - \parbox{10mm}{\includegraphics[scale=1.0]{fig10.eps}}_{(N - 1)} \Biggl) \Biggl] \Biggl\}. 
\end{eqnarray}
After the $\epsilon$-expansion of the $\parbox{7mm}{\includegraphics[scale=.6]{fig10.eps}}_{(1)}$ and $\parbox{7mm}{\includegraphics[scale=.6]{fig10.eps}}_{(N - 1)}$ diagrams exactly in $K$, we have
\begin{eqnarray}\label{gfjifgjigjk}
\parbox{10mm}{\includegraphics[scale=1.0]{fig10.eps}}_{(1)} = \frac{1}{\epsilon}\Biggl[ 1 - \frac{1}{2}\epsilon -\frac{1}{2}L(P^{2} + K_{\mu\nu}P^{\mu}P^{\nu})\epsilon \Biggl]\mathbf{\Pi},
\end{eqnarray}
where
\begin{eqnarray}\label{gjkljjfodfi}
&& L(P^{2} + K_{\mu\nu}P^{\mu}P^{\nu}) = \int_{0}^{1}dx\ln\left[x(1-x)(P^{2} + K_{\mu\nu}P^{\mu}P^{\nu}) + t + \frac{g^{\ast}M^{2}}{2}\right]  
\end{eqnarray}
with an analog expression for $\parbox{7mm}{\includegraphics[scale=.6]{fig10.eps}}_{(N - 1)}$ with $g^{\ast}M^{2}/2 \rightarrow g^{\ast}M^{2}/6$, now defining the correlation length as the second moment of the spin-spin correlation function as $\xi^{2} = (d\Gamma_{L}/d(P^{2} + K_{\mu\nu}P^{\mu}P^{\nu}))|_{P^{2} + K_{\mu\nu}P^{\mu}P^{\nu} = 0}/\Gamma(P^{2} + K_{\mu\nu}P^{\mu}P^{\nu} = 0)$,  
\begin{eqnarray}\label{hjfdhjfhjf}
\xi_{0}^{+} = 1 - \frac{N + 2}{4(N + 8)}\epsilon , 
\end{eqnarray}
\begin{eqnarray}\label{kjlkjhlkh}
\xi_{0}^{-} = 2^{-1/2}\Bigg[1 - \frac{1}{12}\Bigg( \frac{7}{2} +\ln 2 \Bigg)\epsilon\Bigg].
\end{eqnarray}  
As we can see in eq. (\ref{kjlkjhlkh}), $\xi_{0}^{-}$ is not defined for all $N \neq 1$. The reason is the same as for the critical amplitude $C^{-}$: the existence of divergences generated by the presence of Goldstone modes.

\par $R_{\xi}^{+}$. The amplitudes needed for the evaluation of the ratio $R_{\xi}^{+}$ are displayed in eqs. (\ref{hduhsduhsdu}) and (\ref{hjfdhjfhjf}) through the definition $R_{\xi}^{+} = \xi_{0}^{+}(A^{+})^{1/d}$.  

\par $\xi_{0}^{+}/\xi_{0}^{T}$. The momentum-dependent transverse correlation function \cite{BrezinLeGuillouZinnJustin} is given by its diagrammatic expansion \cite{Bervillier} 
\begin{eqnarray}\label{uhdsuhdfuhdfu}
&& \Gamma_{T}(P^{2} + K_{\mu\nu}P^{\mu}P^{\nu},t,M) = P^{2} + K_{\mu\nu}P^{\mu}P^{\nu} + t + \frac{1}{6}g^{\ast}M^{2} + \nonumber \\ && \frac{1}{6}g^{\ast}\Biggl\{ \Biggl[ \parbox{12mm}{\includegraphics[scale=1.0]{fig1.eps}}_{(1)} + (t + g^{\ast}M^{2}/2)\parbox{10mm}{\includegraphics[scale=1.0]{fig10.eps}}_{SP} \Biggl]  + (N + 1)\Biggl[ \parbox{12mm}{\includegraphics[scale=1.0]{fig1.eps}}_{(N - 1)} + (t + g^{\ast}M^{2}/6)\parbox{10mm}{\includegraphics[scale=1.0]{fig10.eps}}_{SP} \Biggl] + \nonumber \\ && \frac{2}{3}g^{\ast}M^{2}\Biggl[ \Biggl(\parbox{10mm}{\includegraphics[scale=1.0]{fig10.eps}}_{SP} - \parbox{10mm}{\includegraphics[scale=1.0]{fig10.eps}}_{(1,(N - 1))} \Biggl) \Biggl] \Biggl\}, 
\end{eqnarray}
where the index $(1,(N - 1))$ in the referred diagram means that its internal propagators are the $\parbox{7mm}{\includegraphics[scale=0.6]{fig9.eps}}_{(1)}^{-1}$ and the $\parbox{7mm}{\includegraphics[scale=0.6]{fig9.eps}}_{(N - 1)}^{-1}$ ones, respectively. The $\epsilon$-expansion of this diagram results in similar expressions to the (\ref{gfjifgjigjk})-(\ref{gjkljjfodfi}) ones with the substitution $ t + g^{\ast}M^{2}/2 \rightarrow x(t + g^{\ast}M^{2}/2) + (1 - x)(t + g^{\ast}M^{2}/6)$. Thus defining the transverse correlation length by $\Gamma_{T}(P^{2} + K_{\mu\nu}P^{\mu}P^{\nu})|_{H = 0}\underset{P \rightarrow 0}\sim (P^{2} + K_{\mu\nu}P^{\mu}P^{\nu})(\xi_{T})^{2 - d}/M^{2}$, we obtain 
\begin{eqnarray}\label{fdgfdsadfs}
\xi_{0}^{T} = \Biggl\{ \frac{\epsilon}{(N + 8)\mathbf{\Pi}}\Bigg[1 + \frac{3}{N + 8}\Bigg( \frac{5}{6} + \ln 2 \Bigg)\epsilon + \Bigg(\frac{9N + 42}{(N + 8)^{2}} -\frac{1}{2} \Bigg)\epsilon \Bigg]\Biggl\}^{1/(d - 2)}. 
\end{eqnarray} 

\par $Q_{2}$. For the evaluation of the ratio $Q_{2}$ defined by $Q_{2} = (C^{+}/C^{c})(\xi_{0}^{c}/\xi_{0}^{+})^{2 - \eta}$ (the critical exponent $\eta$ can be set to zero at the loop level here), we have computed the amplitudes $C^{+}$ and $\xi_{0}^{+}$ already in eqs. (\ref{jhdshudfhuifh}) and (\ref{hjfdhjfhjf}), respectively. Now, from the susceptibility and correlation length at the critical point, we get
\begin{eqnarray}\label{ouypuopuo}
C^{c} = \frac{2D^{1/\delta}}{g^{\ast 1/2\beta}} \Bigg\{ 1 - \frac{9}{2(N + 8)}\Bigg[ \left( 1 - \ln 2 \right) +  \frac{N - 1}{9}\left( 1 - \ln 6 \right) + \frac{2(N + 8)}{27} \Bigg]\epsilon\Bigg\},\nonumber  \\ &&
\end{eqnarray}
\begin{eqnarray}\label{bxvbcvbxvc}
\xi_{0}^{c} = \frac{2^{1/2}D^{1/2\delta}}{g^{\ast 1/4\beta}} \Bigg\{ 1 - \frac{9}{4(N + 8)}\Bigg[ \left( 1 - \ln 2 \right) + \frac{N - 1}{9}\left( 1 - \ln 6 \right) + \frac{N + 14}{27} \Bigg]\epsilon\Bigg\}^{1/2}.
\end{eqnarray}

\par $Q_{3}$. In the definition of the $Q_{3}$ ratio given by $Q_{3} = \widehat{D}(\xi_{0}^{+})^{2 - \eta}/C^{+}$, the amplitudes $C^{+}$ and $\xi_{0}^{+}$ are shown in eqs. (\ref{jhdshudfhuifh}) and (\ref{hjfdhjfhjf}). The remaining $\widehat{D}$ amplitude is evaluated from the momentum-dependent longitudinal correlation function at the critical point and arbitrary momentum. Thus we get
\begin{eqnarray}\label{hjdghssf}
 \widehat{D} = 1. 
\end{eqnarray}
As it can be seen, the rotation symmetry-breaking full $\mathbf{\Pi}$ factor disappear in the final expressions to all the amplitude ratios above and we obtain their rotation-invariant counterparts \cite{PrivmanHohenbergAharony}.

\section{Amplitude ratios in minimal subtraction scheme}

\par In the minimal subtraction renormalization scheme, the external momenta of the $1$PI vertex parts to be renormalized, by minimally subtracting the dimensional poles, are held at arbitrary values, showing that this method is more general and elegant than the earlier. Thus, eqs. (\ref{ygfdxzsze})-(\ref{jijhygtfrd}) must not hold necessarily. Then, minimally renormalizing the bare field $\phi_{B}$, coupling constant $\lambda_{B}$ and composite field coupling constant $t_{B}$, the renormalized effective potential $\mathcal{F}(t,M,g^{\ast})$ of eq. (\ref{fhdldglkjdflk}) assumes a similar form, but with the change $\parbox{6mm}{\includegraphics[scale=0.6]{fig10.eps}}_{SP} \longrightarrow [\parbox{6mm}{\includegraphics[scale=0.6]{fig10.eps}}]_{S}$, where $\parbox{6mm}{\includegraphics[scale=0.6]{fig10.eps}}$ is, for example, the $\parbox{6mm}{\includegraphics[scale=0.6]{fig10.eps}}_{(1)}$ diagram or, similarly, the $\parbox{6mm}{\includegraphics[scale=0.6]{fig10.eps}}_{(N - 1)}$ one with their external momenta held at arbitrary values and $[~~~]_{S}$ means that what is to be considered inside the brackets are the singular terms of the diagram, not the regular ones, as opposed to the normalization conditions renormalization method in which the regular terms are also taken into account. The nontrivial fixed point, i. e., the nontrivial root of the $\beta$-function in this scheme is given by \cite{Int.J.Geom.MethodsMod.Phys.13.1650049}
\begin{eqnarray}\label{gdcnkiy}
&& g^{\ast} = \frac{6\epsilon}{(N + 8)\mathbf{\Pi}}\left[ 1 + \frac{9N + 42}{(N + 8)^{2}} \epsilon \right]. 
\end{eqnarray}
Thus, performing the same steps of the earlier section, the expressions for the amplitudes are  
\begin{eqnarray}\label{terreteu}
&& A^{+} = \frac{N}{4}\Bigg[ 1 + \Bigg( \frac{4}{4 - N} + A_{N}^{\prime} \Bigg)\epsilon \Bigg]\mathbf{\Pi}, 
\end{eqnarray}
\begin{eqnarray}\label{eytwertyer}
 A^{-} = \Bigg[1 + \Bigg( \frac{N}{4 - N} - \frac{4 - N}{2(N + 8)}\ln 2 + A_{N}^{\prime} \Bigg)\epsilon\Bigg]\mathbf{\Pi},
\end{eqnarray}     
where
\begin{eqnarray}\label{reytreytrwe}
A_{N}^{\prime} = - \frac{9N + 42}{(N + 8)^{2}} - \frac{4 - N}{(N + 8)} - \frac{(N + 2)(N^{2} + 30N + 56)}{2(4 - N)(N + 8)^{2}},
\end{eqnarray}
\begin{eqnarray}\label{uiyiupoiu}
&& C^{+} = 1, 
\end{eqnarray}
\begin{eqnarray}\label{iuoipiyo}
&& C^{-} = \frac{1}{2}\Bigg[1 - \frac{1}{6}( 3 + \ln 2 )\epsilon\Bigg],
\end{eqnarray}  
\begin{eqnarray}\label{tretrety}
D = \frac{1}{6}g^{\ast(\delta - 1)/2}\Bigg[1 - \frac{1}{2}\Bigg( \ln 2 + \frac{N - 1}{N + 8}\ln 3 \Bigg)\epsilon\Bigg],
\end{eqnarray} 
\begin{eqnarray}\label{tuytyut}
B = \Bigg\{\frac{N + 8}{\epsilon}\Bigg[1 - \frac{3\ln 2}{N + 8}\epsilon - \frac{9N + 42}{(N + 8)^{2}}\epsilon \Bigg]\mathbf{\Pi}\Bigg\}^{1/2},
\end{eqnarray} 
\begin{eqnarray}\label{xcvbcxzvcxzv}
\xi_{0}^{+} = 1,
\end{eqnarray}
\begin{eqnarray}\label{xvnbxvznbx}
\xi_{0}^{-} = 2^{-1/2}\Bigg[1 - \frac{1}{12}\Bigg( \frac{5}{2} +\ln 2 \Bigg)\epsilon\Bigg],
\end{eqnarray}
\begin{eqnarray}\label{tryteyt}
\xi_{0}^{T} = \Biggl\{ \frac{\epsilon}{(N + 8)\mathbf{\Pi}}\Bigg[1 + \frac{3}{N + 8}\Bigg( -\frac{1}{6} + \ln 2 \Bigg)\epsilon + \frac{9N + 42}{(N + 8)^{2}}\epsilon \Bigg]\Biggl\}^{1/(d - 2)},
\end{eqnarray} 
\begin{eqnarray}\label{bcnbvcbc}
C^{c} = \frac{2D^{1/\delta}}{g^{\ast 1/2\beta}} \Bigg\{ 1 - \frac{9}{2(N + 8)}\Bigg[ - \ln 2 - \frac{N - 1}{9}\ln 6 + \frac{2(N + 8)}{27} \Bigg]\epsilon\Bigg\},
\end{eqnarray}
\begin{eqnarray}\label{bvnbcvb}
\xi_{0}^{c} = \frac{2^{1/2}D^{1/2\delta}}{g^{\ast 1/4\beta}} \Bigg\{ 1 - \frac{9}{4(N + 8)}\Bigg[ - \ln 2 - \frac{N - 1}{9}\ln 6 + \frac{N + 14}{27} \Bigg]\epsilon\Bigg\}^{1/2},
\end{eqnarray}
\begin{eqnarray}\label{eetetee}
&& \widehat{D} = 1. 
\end{eqnarray}
The amplitude ratios obtained in eqs. (\ref{terreteu})-(\ref{eetetee}) are the same as the respective rotation-invariant ones \cite{PrivmanHohenbergAharony}, as well as the equation of state in its universal form \cite{BrezinLeGuillouZinnJustin}. One interesting feature of this method is that the amplitudes themselves do not depend explicitly on the rotation symmetry-breaking full $\mathbf{\Pi}$ factor. The absence of explicit dependence on non-universal features is not an exception of rotation symmetry-breaking properties in this method. It is also observed in treating critical properties of rotation-invariant finite size systems subject to periodic and anti-periodic \cite{SilvaJrLeite} as well as Dirichlet and Newmann \cite{SantosSilvaJrLeite} boundary conditions.

\section{Amplitude ratios in the BPHZ method}\label{Amplitude ratios in the BPHZ method}
 
\par In the BPHZ method, the most general and elegant one, the divergences of the bare theory are eliminated by adding terms to the initially bare Lagrangian density. These terms introduced are called counterterms. This process can be applied order by order in perturbation theory. At the loop level considered in this work, we can once again renormalize the field $\phi_{B}$, coupling constant $\lambda_{B}$ and composite field coupling constant $t_{B}$ by adding counterterms to turn them finite \cite{Int.J.Mod.Phys.B.30.1550259}. The resulting renormalized effective potential $\mathcal{F}(t,M,g^{\ast})$, with $g^{\ast}$ being the same in eq. (\ref{gdcnkiy}), will display once again the form in eq. (\ref{fhdldglkjdflk}), now with $\parbox{6mm}{\includegraphics[scale=0.6]{fig10.eps}}_{SP} \longrightarrow \mathcal{K}(\parbox{6mm}{\includegraphics[scale=0.6]{fig10.eps}})$. The symbol $\mathcal{K}(~~~)$ indicates that only the singular part of the diagram is to be considered. Although this renormalization process be distinct of that of the earlier section, at least at the loop level treated here the effective potential is the same as the one of the earlier section, thus leading to the same amplitudes computed in that section. This leads to identical amplitude ratios found in the last two sections. Now we proceed to compute the amplitude ratios valid for all-loop order.

\section{All-loop order amplitude ratios}\label{All-loop order amplitude ratios}

\par In this section, we show the universality of amplitude ratios for O($N$) rotation symmetry-breaking self-interacting scalar field theories valid for any loop level in the BPHZ method, the most general and elegant one (without loss of generality, similar arguments can be used in the other methods). We follow the same steps of ref. \cite{Bervillier} used for rotation-invariant amplitude ratios. A given general amplitude $A_{G}$ is of the form $A_{G} = X^{a}Y^{b}F(g^{\ast})$, where $a$ and $b$ are integer numbers, critical exponents or a combination of them and $F$, a general function of $g^{\ast}$ only. The $X$ and $Y$ factors are non-universal and are responsible for the scale-dependence of the critical amplitudes. They are related, for example, to the order parameter and conjugate field scales, which are the two fixed independent scales needed for the establishment of the two-scale-factor universality hypothesis. The another source of non-universality is contained in the general function $F(g^{\ast})$, through the fixed point, whose value depends on the renormalization scheme employed. It was shown, in a general proof and therefore valid for any loop level \cite{Bervillier}, that the amplitude ratios are independent of the fixed independent scale factors. Thus, we have, additionally, to show that the amplitude ratios are independent of the rotation symmetry-breaking parameters $K_{\mu\nu}$, thus being identical to the corresponding rotation-invariant ones. This task will be achieved if we show that the rotation symmetry-breaking general function $F(g^{\ast})$ is the same as its rotation-invariant counterpart, i. e. that is does not depend on $K_{\mu\nu}$. We apply the general theorem \cite{CarvalhoSenaJunior} 
\begin{theorem} 
Consider a given Feynman diagram in momentum space of any loop order in a theory represented by the Lagrangian density of Eq. (\ref{huytrji}). Its evaluated expression in dimensional regularization in $d = 4 - \epsilon$ can be written as a general functional $\mathbf{\Pi}^{L}\mathcal{G}(g,P^{2} + K_{\mu\nu}P^{\mu}P^{\nu},\epsilon,\kappa)$ if its rotation-invariant counterpart is given by $\mathcal{G}(g,P^{2},\\ \epsilon,\kappa)$, where $L$ is the number of loops of the corresponding diagram.
\end{theorem}
\begin{proof} 
A general Feynman diagram of loop level $L$ is a multidimensional integral in $L$ distinct and independent momentum integration variables $q_{1}$, $q_{2}$,...,$q_{L}$, each one with volume element given by $d^{d}q_{i}$ ($i = 1, 2,...,L$). As showed in last Section, the substitution $q^{\prime} = \sqrt{\mathbb{I} + \mathbb{K}}\hspace{1mm}q$ transforms each volume element as $d^{d}q^{\prime} = \sqrt{det(\mathbb{I} + \mathbb{K})}d^{d}q$. Thus $d^{d}q = d^{d}q^{\prime}/\sqrt{det(\mathbb{I} + \mathbb{K})} \equiv \mathbf{\Pi}d^{d}q^{\prime}$, $\mathbf{\Pi} = 1/\sqrt{det(\mathbb{I} + \mathbb{K})}$. Then, the integration in $L$ variables results in a rotation symmetry-breaking overall factor of $\mathbf{\Pi}^{L}$. Now making $q^{\prime} \rightarrow P^{\prime}$ in the substitution above, where $P^{\prime}$ is the transformed external momentum, then $P^{\prime 2} = P^{2} + K_{\mu\nu}P^{\mu}P^{\nu}$. So a given Feynman diagram, evaluated in dimensional regularization in $d = 4 - \epsilon$, assumes the expression $\mathbf{\Pi}^{L}\mathcal{G}(g,P^{2} + K_{\mu\nu}P^{\mu}P^{\nu},\epsilon,\kappa)$, where $\mathcal{G}$ is associated to the corresponding diagram if the rotation-invariant Feynman diagram counterpart evaluation results in $\mathcal{G}(g,P^{2},\epsilon,\kappa)$.
\end{proof}

Now we can write a general Feynman diagram in the form $\mathbf{\Pi}^{L}\mathcal{G}(g,P^{2} + K_{\mu\nu}P^{\mu}P^{\nu},\epsilon,\kappa)$ if its rotation-invariant counterpart is given by $\mathcal{G}(g,P^{2},\epsilon,\kappa)$, where $L$ is the number of loops of referred diagram. Thus, as the general function depends only on $g$ and at a given of its terms each term of one order of $g$ is also a term of one order of loop, it is given by $F(\mathbf{\Pi} g^{\ast})$. The all-loop rotation symmetry-breaking fixed point $g^{\ast}$, taking into account the rotation symmetry-breaking breaking mechanism exactly, is related to its rotation-invariant counterpart $g^{\ast(0)}$ through $g^{\ast} = g^{\ast(0)}/\mathbf{\Pi}$ \cite{CarvalhoSenaJunior}. Now, substituting $g^{\ast} = g^{\ast(0)}/\mathbf{\Pi}$, we have that $F(\mathbf{\Pi} g^{\ast}) \equiv F(g^{\ast(0)})$. Then, the rotation symmetry-breaking general function is the same as the corresponding rotation-invariant one, leading to the rotation symmetry-breaking amplitude ratios identical to their rotation-invariant counterparts. This completes our task.

\section{Conclusions}

\par In this paper, we have evaluated the all-loop quantum corrections to the amplitude ratios for rotation symmetry-breaking O($N$) $\lambda\phi^{4}$ scalar field theories, taking exactly the rotation symmetry-breaking mechanism into account, through coordinates redefinition applied in momentum space directly in Feynman diagrams, and thus avoiding the tedious calculation in powers of $K_{\mu\nu}$, by employing field-theoretic renormalization group, dimensional regularization and $\epsilon$-expansion techniques in three distinct methods. Firstly, we have explicitly computed analytically the amplitude ratios at the one-loop level and finally, in a proof by induction through a general theorem emerging from the exact calculation, computed the quantum corrections for any loop level. We have showed that, although a same rotation symmetry-breaking critical amplitude can be different in distinct methods and thus can dependent on the renormalization method employed, the outcome for the amplitude ratios are the same and, furthermore, equal to their rotation-invariant counterparts. This result reveals that the amplitude ratios do not depend on the renormalization scheme employed and on the exact rotation symmetry-breaking mechanism, thus being universal quantities and ratifying the robustness of the O($N$) two-scale-factor universality hypothesis. We can interpret physically this result by realizing that the symmetry breaking mechanism does not occur in the internal symmetry space of the field, but in the one in which the field is defined \cite{Aharony}. Then, if the amplitude ratios are really to be universal quantities, the values of the amplitude ratios must be not affected by this symmetry breaking mechanism and they depend just, as usual, on the $d$ and $N$ parameters and the symmetry of the $N$-component order parameter. Furthermore, this work can shed light on the understanding of the exact rotation symmetry-breaking properties of corrections to scaling and finite-size scaling of rotation symmetry-breaking amplitude ratios as well as critical exponents in geometries subjected to different boundary conditions for systems undergoing a continuous phase transition for the systems studied here as well as for anisotropic ones \cite{Leite104415,Leite224432,CarvalhoLeite178,CarvalhoLeite151}.

\section{Acknowledgements}

\par JFSN, KALL and PRSC would like to thank FAPEPI-CAPES, CAPES (Brazilian funding agencies) and Universidade Federal do Piau\'{i} for financial support, respectively. MISJ was financially supported by FAPEAL (Alagoas State Research Foundation) and CNPq (Brazilian Funding Agency).

\end{document}